\documentclass[12pt]{amsart} 

\usepackage{amsmath,amssymb}
\usepackage{amsthm}
\usepackage[nobysame]{amsrefs}
\usepackage{indentfirst}
\usepackage{mathrsfs}
\usepackage{graphicx}
\usepackage{hyperref}
\usepackage{xcolor}
\usepackage{fourier}
\usepackage{geometry}

\definecolor{brown}{rgb}{0.88,.5,0}

\theoremstyle{plain}

\newtheorem{lemma}{Lemma}

\theoremstyle{definition}

\theoremstyle{Remark}
\newtheorem{remark}{Remark}

\DeclareMathOperator{\divop}{div}

\newcommand{\ud}{\,\mathrm{d}}

\newcommand{\NN}{\mathbb{N}}
\newcommand{\ZZ}{\mathbb{Z}}

\newcommand{\wt}[1]{\widetilde{#1}}
\newcommand{\wh}[1]{\widehat{#1}}

\newcommand{\veps}{\varepsilon}
\newcommand{\eps}{\epsilon}

\newcommand{\abs}[1]{\lvert#1\rvert}

\newcommand{\average}[1]{\left\langle#1\right\rangle}

\newcommand{\CB}{\mathrm{CB}}

\begin{document}

\title{Traction Boundary Conditions for Molecular Static Simulations}

\author{Xiantao Li} \address{Department of Mathematics, The
  Pennsylvania State University, University Park, PA 16802}

\author{Jianfeng Lu} \address{Departments of Mathematics, Physics, and
  Chemistry, Duke University, Durham, NC 27708}
\begin{abstract}
This paper presents a consistent approach to prescribe traction boundary conditions in atomistic models. 
Due to the typical multiple-neighbor interactions, finding an appropriate boundary condition that models a desired traction is a non-trivial task. We first present
a one-dimensional example, which demonstrates how such boundary conditions can be formulated. We further analyze  the stability, and derive its continuum limit. We also show how the boundary conditions can be extended to higher dimensions with an application to a dislocation dipole problem under shear stress.      
\end{abstract}


\maketitle

\section{Introduction}

Atomistic models have established a critical role in material modeling
and simulations. In order to study the mechanical responses, boundary
conditions (BC) must be imposed.  While specifying the displacement of
the atoms at the boundary is straightforward, imposing a traction is
much more challenging due to the fact that the range of the atomic
interactions typically goes beyond nearest neighbors.  In direct
contrast to continuum mechanics, where the boundary is of lower
dimension (curves or surfaces), the `boundary' in an atomistic model
often consists of a few layers of atoms. As a result, there are
multiple ways to prescribe a traction BC. For instance, forces can be
applied directly to atoms at the boundary in such a way that they add
up to the given traction. However, it is unclear how to distribute
these forces among the atoms. In particular, boundary layers may
develop and create large modeling error.

Meanwhile, many mathematical problems associated with material defects have been formulated as a system 
under traction. Examples include cracks under mode-I loading  \cite{sih1968fracture}, where uniform stress 
can be specified in the far field, and dislocations under shear stress \cite{hirth1982theory}, which led to the important concepts
of Peierls barrier. Problems of this type can not simply be treated with BCs that prescribe the displacement of the atoms at the boundary. Another possible approach to introduce traction is the Parrinello-Rahman method \cite{parrinello1981polymorphic}, where the stress is created by allowing the shape of the simulation cell to 
change, which is particularly useful when phase transformation processes occur. But the method is limited to periodic cells, and it can not treat material defects without introducing artificial images. 

The purpose of this paper is to formulate a proper BC that represents a traction force along 
the boundary. We set up the problem by embedding the computational domain within an infinite molecular system, where the traction
in the far-field can be introduced. This is motivated by the observation that molecular simulations are 
typically conducted within part of the entire sample, due to the heavy computational cost. Mathematically, the extra degrees of freedom in the surrounding region can be eliminated
by solving the finite difference equations associated with the molecular statics model. This gives rise to 
a BC, which is expressed as an extrapolation of the displacement to the atoms outside the boundary,
along with a shift vector, which depends on the traction in the far field. We further demonstrate that the typical approach in 
which external forces are directly applied at the boundary might be incompatible with these BCs, and that they can lead to
ill-posed problems. 

The present approach allows one to simulate a material system with local defects under traction load, which mimics a surround
elastic medium.
Another potential application is to the domain decomposition (DD) 
method for solving a large-scale molecular system, where the problem is divided into sub-problems, each of which is associate with
a sub-domain.  In particular, the Dirichlet-Neumann method and Neumann-Neumann method (e.g., see \cite{toselli2005domain}) offer a coupling strategy without creating overlapping regions. Our method can be implemented within the DD framework to fascilitate parallelization. 

The paper is organized as follows: We first consider a one-dimensional system to demonstrate how the BC can be derived. We further analyze the stability of the resulting boundary value problem and the continuum limit. As an application, and a demonstration of such BCs in high dimensions, we consider a dislocation dipole problem in section \ref{sec: 2d}. We close the paper by a summary and some discussions.

\section{A one-dimensional example}

To better illustrate the idea, let us first consider a one-dimensional
semi-infinite chain of atoms with undeformed position $x_i = i$, where
$i \in \ZZ_+ \cup \{0\}$. We will also use the undeformed position to
label the atoms. The deformed position is denoted by $y_i$ with displacement $u_i$. We assume
that the atomistic potential has next-nearest-neighbor pairwise interactions. Namely, the total energy in the bulk can 
be written as,
\[ E= \sum_{i \geq 0} \bigl( V(y_{i+2} - y_{i}) +
  V(y_{i+1} - y_{i})  \bigr).\] 
The extension to more general potentials and higher dimensions will be
discussed later.

Intuitively, there are at least two ways to impose a traction at the
boundary. For instance, one may apply forces, denoted by $T_0$ and $T_1$, to the
first two atoms, located at $y_0$ and $y_1,$ respectively.  Alternatively, one
can introduce two additional atoms outside the boundary, which in the
present case, have current positions $y_{-2}$ and $y_{-1}$. These additional atoms will be referred to as {\it ghost atoms},
since they play different roles as the atoms in
the interior.  By specifying $y_{-2}$ and $y_{-1}$ (or $u_{-2}$ and $u_{-1}$), one also creates a traction at the boundary. These two
methods are illustrated in Fig.~\ref{fig:chain}. 
\begin{figure}[ht]
\begin{center}
\hspace{2cm}\scalebox{0.45}{\includegraphics{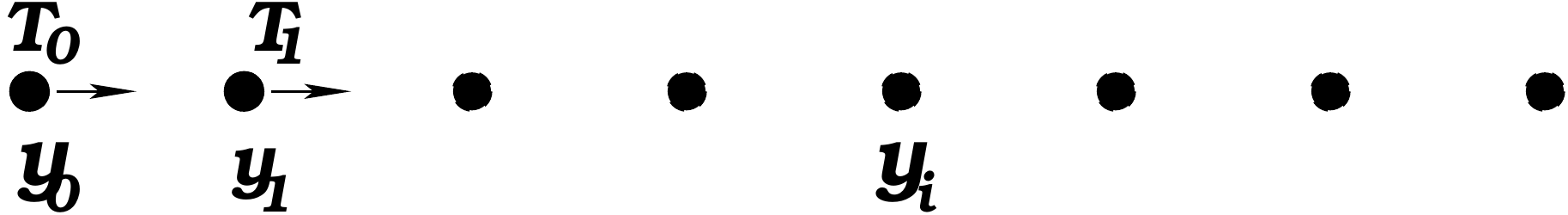}}
\end{center}
\vspace{0.5cm}
\begin{center}
\scalebox{0.45}{\includegraphics{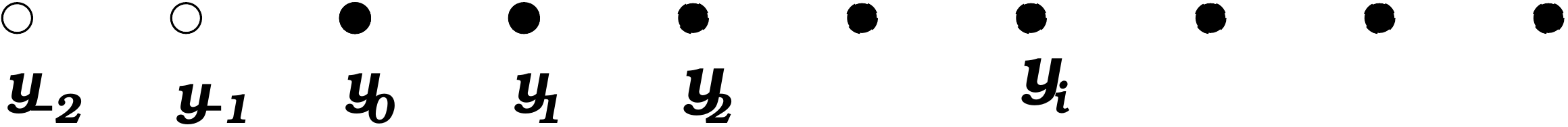}}
\caption{A one-dimensional chain with next nearest-neighbor interactions.
Top:  Tractions, in the form of point forces, are applied to the two atoms at the
left boundary; Bottom: two ghost atoms  are introduced outside the boundary. \label{fig:chain}
}
\end{center}
\end{figure}
\begin{remark}
At this point, both methods seem plausible, and it is not immediately
clear how they are related to each other. Also notable is that there
are {\it two} parameters available to specific {\it one} traction
condition. This will be clarified in the next section.
\end{remark}

We illustrate possible BCs with the following numerical experiment: We
consider the one-dimensional chain model with harmonic
interactions. The force constants are $\kappa_1=1$ and $\kappa_2=-0.2$
for the nearest neighbor and next nearest neighbor interactions. A
traction force needs to be applied at the left boundary, while the
atoms at $x_N$ and $x_{N+1}$ are fixed. We choose $N=20.$ Three
traction BCs are tested. In the first case, we apply a unit force on
the first atom, and in the second case, we apply the same force to the
second atom. In the third test, we split the force among the first two
atoms ($\frac12$ and $\frac12$). For such a simple setup, one would
anticipate that the corresponding continuum model has a simple
solution which is given by a uniform deformation gradient. These
results are shown in Fig. \ref{fig:4bcs}. In all these cases, the
solutions develop a boundary layer, and none of them is fully
consistent with the continuum solution.  As comparison, we include the
result from the traction BC that will be derived later, in which the
position of the first two atoms is determined based on the given
traction.  It is clear that the boundary layer has been eliminated.
\begin{figure}[ht]
\begin{center}
{\includegraphics[width=12cm,height=6cm]{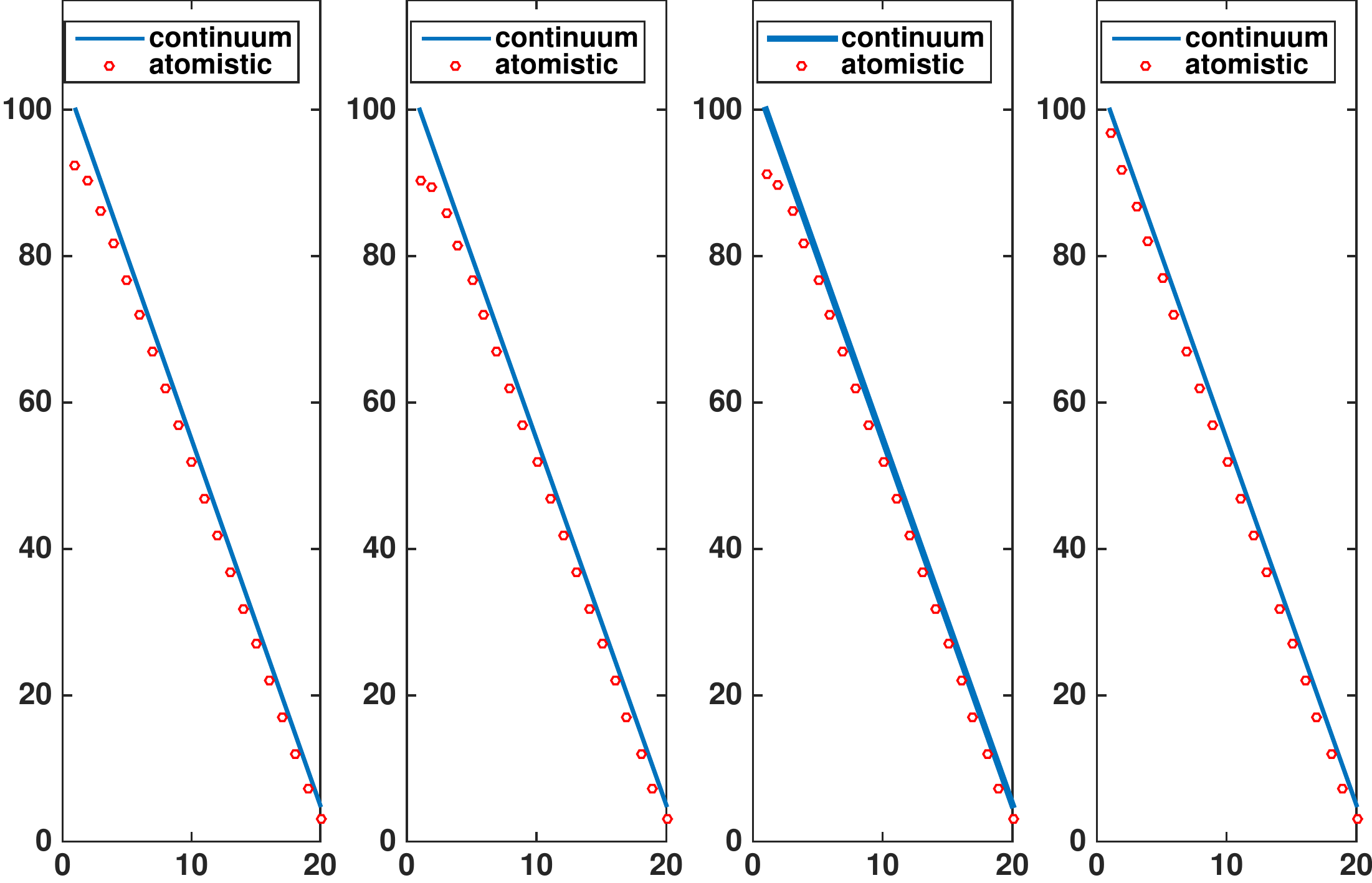}}
\caption{Testing several boundary conditions for the one-dimensional model. 
Left to right: A unit force applied to the first atom; A unit force applied to the second atoms; Forces applied to both atoms; 
Traction boundary condition that will be derived in the next section.
\label{fig:4bcs}
}
\end{center}
\end{figure}

\subsection{The derivation of the traction BC}

To understand the traction BC, we start  by embedding
the semi-infinite atom chain into an infinite chain. We recall that $x_i =
i$, $i \in \ZZ$ denotes the equilibrium positions. We take the view point that
the BCs acting on the atom at $x_0$ should be determined by
the interaction of the semi-infinite chain with the atoms on the left ($i<0$),
whose degrees of freedom will be implicitly incorporated. In other words, the
atoms $x_i$, $i \in \ZZ_-$ serves as an environment for the system we
consider. We will hence distinguish the two groups of atoms by referring them as system atoms and environment
atoms, respectively. This is based on our observation that most atomistic simulations are focused on
part of the entire sample due to the small spatial scale associated with molecular models.  

The problem now is reduced to removing the atoms in the environment. 
To facilitate the reduction of degrees of freedom from the whole chain
to the semi-infinite, we will take the harmonic approximation. This
amounts to assuming that the interaction between the atoms $x_i$ and
$x_j$ is harmonic if either $i < 0$ or $j < 0$. In this paper, we focus our attention to static problems, which can
be formulated as an energy minimization problem. For the current problem,  the potential energy will be divided into
several terms in accordance with the partition of the system,
\begin{equation}
  E = E_{\text{sys}} + E_{\text{int}} + E_{\text{env}} 
\end{equation} 
where $E_{\text{sys}}$ is the interaction among the atoms in the
semi-infinite chain on the right: 
\begin{equation}
  E_{\text{sys}} = \sum_{i \geq 0} \bigl( V(y_{i+2} - y_{i}) +
  V(y_{i+1} - y_{i})  \bigr). 
\end{equation}
Meanwhile $E_{\text{int}}$ collects the interaction terms between a system atom
and an environmental atom. In terms of the displacement $u_j = y_j -
x_j$, we have,
\begin{equation}
  E_{\text{int}} = \frac{\kappa_1}{2} (u_0 - u_{-1})^2 + \frac{\kappa_2}{2} (u_0 - u_{-2})^2 + \frac{\kappa_2}{2} (u_1 - u_{-1})^2, 
\end{equation}
where $\kappa_1 = V''(1)$ and $\kappa_2 = V''(2)$ are two stiffness
constants for nearest-neighbor and next-nearest-neighbor
interactions. Further, $E_{\text{env}}$ denotes the energy for the environment,
given by
\begin{equation}
  E_{\text{env}} = \sum_{i < 0} \bigl[ \frac{\kappa_1}{2} (u_i - u_{i-1})^2  + \frac{\kappa_2}{2} (u_i - u_{i-2})^2 \bigr]. 
\end{equation}

Since the interaction is assumed to be next-nearest-neighbor, the
environment atoms only interact with atoms with reference positions $x_0$ and $x_1$, but not 
other system atoms. Given the displacement $u_0$ and $u_1$, the force balance equations for the environment atoms can be written as 
\begin{equation}\label{eq:fdhalfchain}
  \kappa_2 (u_{j+2} - 2 u_j + u_{j-2}) + \kappa_1 (u_{j+1} - 2 u_j + u_{j-1}) = 0, 
  \qquad j \in \ZZ_-. 
\end{equation}
The general solution of this finite difference equation is given by, 
\begin{equation}
  u_j = A + B j + C \lambda^j + D \lambda^{-j}, 
\end{equation}
where $\lambda$ is a root of the characteristic polynomial associated
to \eqref{eq:fdhalfchain} (the other roots are $1/\lambda$ and $1$
with multiplicity $2$):
\begin{equation}\label{eq:lambda}
  \lambda = -1 - \dfrac{\kappa_1}{2\kappa_2} \left [1 - \sqrt{1+ \frac{4\kappa_2}{\kappa_1}}\,\right].
\end{equation}
We collect some elementary properties of $\lambda$ in the following lemma.
\begin{lemma}\label{lem:lambda}
  Assume $\kappa_1 > 0$ and $\kappa = \kappa_1 + 4 \kappa_2 > 0$, we
  have $\abs{\lambda} \leq 1$ and $-\kappa_2 \lambda \geq 0$.
\end{lemma}
\begin{proof} 
  First notice that 
  \begin{equation*}
    1 + \frac{4 \kappa_2}{\kappa_1} \leq 
    1 + \frac{4 \kappa_2}{\kappa_1} + \frac{4 \kappa_2^2}{\kappa_1^2}= \Bigl( 1 + \frac{2 \kappa_2}{\kappa_1}\Bigr)^2, 
  \end{equation*}
  which yields
  \begin{equation*}
    1 - \sqrt{1 + \frac{4 \kappa_2}{\kappa_1}} \geq - \frac{2 \kappa_2}{\kappa_1}.
  \end{equation*}

  For $\kappa_2 > 0$, we then get
  \begin{equation*}
    0 \leq - \frac{\kappa_1}{2 \kappa_2} \biggl( 1 - \sqrt{1 + \frac{4 \kappa_2}{\kappa_1}} \biggr) \leq 1.
  \end{equation*}
  Hence, by definition of $\lambda$ in \eqref{eq:lambda}, we get
  $\lambda \in [-1, 0]$.

  For $\kappa_2 < 0$, we have
  \begin{equation*}
    - \frac{\kappa_1}{2 \kappa_2} \biggl( 1 - \sqrt{1 + \frac{4 \kappa_2}{\kappa_1}} \biggr) \geq 1.
  \end{equation*}
  This yields $\lambda \geq 0$. Since $1 + 4 \kappa_2 / \kappa_1 < 1$, we also have 
  \begin{equation*}
    \sqrt{1 + \frac{4 \kappa_2}{\kappa_1}} \geq 1 + \frac{4\kappa_2}{\kappa_1}
  \end{equation*}
  and hence 
  \begin{equation*}
    - \frac{\kappa_1}{2 \kappa_2} \biggl( 1 - \sqrt{1 + \frac{4 \kappa_2}{\kappa_1}} \biggr) \leq 
    - \frac{\kappa_1}{2 \kappa_2} \biggl( 1 - 1 - \frac{4 \kappa_2}{\kappa_1} \biggr) = 2. 
  \end{equation*}
  We conclude that $\lambda \leq 1$.
\end{proof}
Since $\abs{\lambda} < 1$, $\lambda^{j}$ grows exponentially as $j \to
-\infty$ and hence unphysical. This leads to the requirement that $C =
0$.  The positions of the atoms at $x_0$ and $x_1$ further provide two
BCs for \eqref{eq:fdhalfchain}. The remaining one degree of freedom is
determined by the traction at the boundary,
\begin{equation}\label{eq:tractioncondition}
  - \kappa_1 (u_{0} - u_{-1}) - \kappa_2 (u_{0} - u_{-2}) - \kappa_2 (u_{1} - u_{-1}) = T, 
\end{equation}
where $T$ is the prescribed traction at the boundary (scalar in $1D$).  Intuitively, the traction across a material
interface is given by the sum of the forces between two atoms that are on different sides of the interface \cite{Admal2010unified,WuLi14}. This is indeed non-trivial, especially for multi-body interactions. But formulas are available for most empirical potentials \cite{WuLi14}. 
\begin{remark}
  We further remark that the traction is conserved since no external
  force acts on the fictitious atoms: For $j \in \ZZ_-$,
  \begin{multline*}
    \Bigl(- \kappa_1 (u_{j} - u_{j-1}) - \kappa_2 (u_{j} - u_{j-2})
    \Bigr) - \kappa_2 (u_{j+1} - u_{j-1}) \\
    \stackrel{\eqref{eq:fdhalfchain}}{=}
    \Bigl(- \kappa_1 (u_{j+1} - u_{j}) - \kappa_2 (u_{j+2} - u_{j}) \Bigr) - \kappa_2 (u_{j+1} - u_{j-1}) \\
    = - \kappa_1 (u_{j+1} - u_{j}) - \kappa_2 (u_{j+1} - u_{j-1}) -
    \kappa_2 (u_{j+2} - u_{j}).
  \end{multline*}
\end{remark}

The solution to \eqref{eq:fdhalfchain} can now be found. In
particular, the coefficients are
\begin{equation*}
  A = \frac{-\lambda^{-1} u_0 + u_1 + T / \kappa}{1 - \lambda^{-1}}, \quad B = - T / \kappa, \quad \text{and} \quad C = \frac{u_0 - u_1 - T/ \kappa}{1 - \lambda^{-1}},
\end{equation*}
where $\kappa = \kappa_1 + 4 \kappa_2$. As a result, the
displacements $u_{-1}$ and $u_{-2}$ are given by 
\begin{subequations}\label{eq:tbc1d}
\begin{align}
  & u_{-1} = (1+ \lambda) u_0 - \lambda u_1 + (1  - \lambda) T / \kappa;  \\
  & u_{-2} = (1+ \lambda) u_{-1} - \lambda u_0 +
  (1 - \lambda) T / \kappa.
\end{align}
\end{subequations}
Notice that $u_{-1}$ and $u_{-2}$ depend linearly on the displacements
$u_0$, $u_1$ and the traction $T$. This is a result of the harmonic
approximation in the environment. Now that the degrees of freedom associated
with the atoms further on the left are removed, we can  formulate the boundary
value problem for the semi-infinite atom chain $x_i$, $i \in \ZZ_+ \cup
\{0\}$ in terms of ghost atoms $x_{-1}$ and $x_{-2}$ at the
boundary. In general, the number of needed layers of atoms is determined
by the interaction range.

With the BCs, the molecular statics model is complete. We 
consider a slightly more general problem by allowing body forces $f_i$ to be applied 
to the system atoms.  In this case,  the force balance equations read as follows,
\begin{subequations}\label{eq:forcebalance}
\begin{align}
  & - V'(y_{j+2} - y_j) - V'(y_{j+1} - y_j) + V'(y_j - y_{j-1}) +
  V'(y_j - y_{j-2}) = f_j, \qquad j \geq 2 \\
  & - V'(y_3 - y_1) - V'(y_2 - y_1) + V'(y_1 - y_0) +
  \kappa_2 (u_1 - u_{-1}) = f_1, \\
  & - V'(y_2 - y_0) - V'(y_1 - y_0) + \kappa_1 (u_0 - u_{-1}) +
  \kappa_2 (u_0 - u_{-2}) = f_0,
\end{align}
\end{subequations}
together with the BCs given by \eqref{eq:tbc1d}. 

Another observation is  that due to the semi-infinite nature,
\eqref{eq:forcebalance}-\eqref{eq:tbc1d} can only determine $u$ up to
a constant. To uniquely fix the arbitrary constant, we choose $u_0 =
0$. In addition, while the solution $u$ can be a linear function
that corresponds to a uniformly stretched (or compressed) state, it is natural to
exclude those solutions that grow superlinearly at infinity \cite{luskin2013atomistic}. Hence,
the complete set of BCs for the semi-infinite chain
consists of the traction BC \eqref{eq:tbc1d} and the conditions 
\begin{subequations}\label{eq:infbc1d}
\begin{align}
  & u_0 = 0; \\
  & \limsup_{j \to \infty} \frac{\abs{u_j}}{j} < \infty. 
\end{align}
\end{subequations}

We emphasize that the above two BCs are associated
with the {\it semi-infinite} chain under consideration, but not the
traction at the left end of the chain. For finite system with a right
boundary, appropriate BCs should be chosen to replace
\eqref{eq:infbc1d} according to the physical situation. Our emphasis,
however, is on the traction condition at the left boundary.

Let us summarize the general framework for our BC construction in one-dimensional systems as follows 
\begin{itemize}
\item[Step 1.] Supplement the system with a fictitious environment of
  atoms with linear approximation; 
\item[Step 2.] Solve the positions of the environmental atoms with the
  condition of fixed traction; 
\item[Step 3.] The BC of the atomistic system is then
  given in terms of the positions of the ghost atoms.
\end{itemize}
This procedure can be clearly generalized to one-dimensional atomistic systems
with arbitrary short-range interactions. The number of BCs depends on the interaction range.

Next, we turn to several properties of the BCs. These
will help us better understand the traction BC and
also facilitate the extension to higher dimension.

\subsection{The continuum limit}\label{sec:contlimit}
For continuum elasticity models, traction BCs are imposed
as the normal component of the stress.  In this subsection, we show that the continuum
limit of the reduced system \eqref{eq:forcebalance}, together with the
BCs \eqref{eq:tbc1d}, leads to the Cauchy-Born
elasticity with continuum traction BC in
elasticity. Hence, our BCs \eqref{eq:tbc1d} can be
viewed as the molecular statics analog of the traction boundary
condition in continuum elasticity.

To this end, we adopt the natural rescaling of the
system such that the distance between nearest-neighbor atoms in
equilibrium becomes $\veps$. We will use superscript to make explicit
the dependence on the scaling parameter $\veps$.  Hence, the
equilibrium positions scale to $x_j^{\veps} = j \veps$, $j \in \ZZ_+
\cup \{0\}$ and the deformed positions are $y_j^{\veps} = x_j^{\veps}
+ u^{\veps}(x_j^{\veps})$. We rewrite the force balance equation and
the traction BC accordingly:
\begin{subequations}\label{eq:forcebalanceeps}
  \begin{align}
    & - V'\Bigl(\frac{y_{j+2}^{\veps} - y_j^{\veps}}{\veps}\Bigr) -
    V'\Bigl(\frac{y_{j+1}^{\veps} - y_j^{\veps}}{\veps}\Bigr) +
    V'\Bigl(\frac{y_j^{\veps} - y_{j-1}^{\veps}}{\veps}\Bigr) +
    V'\Bigl(\frac{y_j^{\veps} - y_{j-2}^{\veps}}{\veps}\Bigr) = \veps
    f_j, \qquad j \geq 2
    \label{eq:atomj}\\
    & - V'\Bigl(\frac{y_3^{\veps} - y_1^{\veps}}{\veps}\Bigr) -
    V'\Bigl(\frac{y_2^{\veps} - y_1^{\veps}}{\veps}\Bigr) +
    V'\Bigl(\frac{y_1^{\veps} - y_0^{\veps}}{\veps}\Bigr) + \kappa_2
    \Bigl(\frac{u_1^{\veps} - u_{-1}^{\veps}}{\veps}\Bigr) = \veps
    f_1,
    \label{eq:atom1} \\
    & - V'\Bigl(\frac{y_2^{\veps} - y_0^{\veps}}{\veps}\Bigr) -
    V'\Bigl(\frac{y_1^{\veps} - y_0^{\veps}}{\veps}\Bigr) + \kappa_1
    \Bigl(\frac{u_0^{\veps} - u_{-1}^{\veps}}{\veps}\Bigr) + \kappa_2
    \Bigl(\frac{u_0^{\veps} - u_{-2}^{\veps}}{\veps}\Bigr) = \veps f_0
    \label{eq:atom0}
  \end{align}
\end{subequations}
with 
\begin{subequations}\label{eq:tbc1deps}
  \begin{align}
    & u_{-1}^{\veps} = (1+ \lambda) u_0^{\veps} - \lambda u_1^{\veps} + \veps (1 - \lambda) T / \kappa;  \\
    & u_{-2}^{\veps} = (1+ \lambda) u_{-1}^{\veps} - \lambda
    u_0^{\veps} + \veps (1 - \lambda) T / \kappa.
  \end{align}
\end{subequations}
We now take the continuum limit $\veps \to 0$. To the leading order, the
equation \eqref{eq:atomj} becomes
\begin{equation}\label{eq:contforcebalance}
  - \divop \bigl[ V'( y'(x)) + 2 V'(2 y'(x))\bigr] = f(x).
\end{equation}
Note that for the current atomistic interaction potential, the
Cauchy-Born stored energy density is given by \cite{EMing:2007a}:
\begin{equation}
  W_{\CB}(A) = V\bigl( I + A \bigr) +  V\bigl( 2 (I + A) \bigr). 
\end{equation}
Hence, \eqref{eq:contforcebalance} is exactly the force balance
equation for the Cauchy-Born elasticity since
\begin{equation}
  \partial_A W_{\CB}(u') = V'\bigl( 1 + u'(x) \bigr) + 2 V'\bigl( 2 + 2 u'(x)\bigr).
\end{equation}

Combining \eqref{eq:atom1}, \eqref{eq:atom0} and the BCs \eqref{eq:tbc1deps}, we get 
\begin{equation}
  - V'\Bigl(\frac{y_3^{\veps} - y_1^{\veps}}{\veps}\Bigr) - V'\Bigl(\frac{y_2^{\veps} -
    y_1^{\veps}}{\veps}\Bigr)  - V'\Bigl(\frac{y_2^{\veps} - y_0^{\veps}}{\veps}\Bigr) 
  = \veps f_1 + \veps f_0 + T. 
\end{equation}
To the leading order, this yields 
\begin{equation}\label{eq:contboundarycond}
  - V'\bigl( 1 + u'(0) \bigr) - 2 V'\bigl( 2 + 2 u'(0)\bigr) = T. 
\end{equation}
As the left hand side is equal to $n \cdot \partial_A W_{\CB}(u') \big
\vert_{x = 0}$, where $n$ is the unit exterior normal, the BC \eqref{eq:contboundarycond} is exactly the traction BC for the elastic energy density $W_{\CB}$.

\subsection{Linear stability at the equilibrium}
We also make the observe that the force balance equations \eqref{eq:forcebalanceeps}
with the traction BC \eqref{eq:tbc1deps} can be viewed
as a finite difference system with BCs. It is then
natural to analyze the stability of such a finite difference system,
similar in spirit to the analysis in the context of
atomistic-to-continuum methods \cite{LuMing:13, LuMing:14}.  We also
note that the stability is the crucial ingredient for the rigorous
proof of the continuum limit in \S\ref{sec:contlimit}.  The stability
of molecular statics models under periodic and Dirichlet BCs have been analyzed in \cite{EMing:2007a,ehrlacher2013analysis}.

To understand the stability, we linearize the force balance equations
\eqref{eq:forcebalance} at the equilibrium (undeformed) state, yielding,
\begin{equation}\label{eq:linearforcebalance}
  - \kappa_2 (u_{j+2} - 2u_j + u_{j-2}) -
  \kappa_1 (u_{j+1} - 2 u_j + u_{j-1}) = f_j, \qquad j \geq 0, 
\end{equation}
supplemented by the BC \eqref{eq:tbc1d} and
\eqref{eq:infbc1d}. Given $T$, we define the map $H_{T}: l^2(\NN) \to
l^2(\NN)$ as
\begin{equation}
  (H_{T} u)_j = - \kappa_2 (u_{j+2} - 2u_j + u_{j-2}) -
  \kappa_1 (u_{j+1} - 2 u_j + u_{j-1}), \qquad j \geq 0 
\end{equation}
with $u_{-1}$ and $u_{-2}$ determined by \eqref{eq:tbc1deps} (and
hence the dependence on $T$). Thus we have
\begin{equation}
  H_T(u)_j - H_0(u)_j = 
  \begin{cases}
    - \bigl(\kappa_2 ( 2 + \lambda ) + \kappa_1 \bigr)
    (1 - \lambda) T / \kappa,  & j = 0; \\
    - \kappa_2 (1 - \lambda) T / \kappa, & j = 1; \\
    0, & \text{otherwise}.
  \end{cases}
\end{equation}

Let us introduce a few short-hand notations. First we define the forward
difference as
\begin{equation}
  (D u)_j =  u_{j+1} - u_j. 
\end{equation}
Moreover, the discrete Laplacian is given by 
\begin{equation}
  (\Delta_d u)_j = 
  u_{j+1} - 2 u_j + u_{j-1}.
\end{equation}
Direct calculations yield,
\begin{multline}
  (\Delta_d \Delta_d u)_j + 4 (\Delta_d u)_j =
  \bigl( u_{j+2} - 4 u_{j+1} + 6 u_j - 4 u_{j-1} + u_{j-2} \bigr) +  
  4 \bigl(u_{j+1} - 2 u_j + u_{j-1}\bigr) \\
  = ( u_{j+2} - 2 u_j + u_{j-2}).
\end{multline}
With these preparations, we calculate the quadratic form $\average{u, H_0(u)}$:
\begin{equation}
  \begin{aligned}
    \average{u, H_0 u} & = - \kappa_2 \sum_{j=0}^{\infty} u_j \bigl(
    (\Delta_d \Delta_d u)_j + 4 (\Delta_d u)_j \bigr) - \kappa_1
    \sum_{j=0}^{\infty} u_j (\Delta u)_j \\
    & = - \kappa_2 \sum_{j=0}^{\infty} u_j (\Delta \Delta u)_j -
    \kappa \sum_{j=0}^{\infty} u_j (\Delta u)_j
  \end{aligned}
\end{equation}
Summation by parts gives (recall that $u_0 = 0$)
\begin{equation}
  \sum_{j=0}^{\infty} u_j (\Delta u)_j = - \sum_{j=0}^{\infty} \abs{ (Du)_j}^2.
\end{equation}
For the fourth order term, we get 
\begin{equation}
  \sum_{j=0}^{\infty} u_j (\Delta \Delta u)_j 
  = \sum_{j=0}^{\infty} \abs{ (\Delta u)_j }^2 - u_{-1} (u_1 + u_{-1}) 
  = \sum_{j=0}^{\infty} \abs{ (\Delta u)_j }^2 + \lambda(1 - \lambda) u_1^2,
\end{equation}
where in the last equality, we used $u_{-1} = - \lambda u_1$ in the
case $T = 0$ and $u_0 = 0$.

Finally, we have,  
\begin{equation}
  \average{u, H_0 u} = \kappa \sum_{j=0}^{\infty} \abs{ ( D u)_j}^2 - \kappa_2 \sum_{j=0}^{\infty} \abs{ (\Delta u)_j }^2 - \kappa_2 \lambda (1 - \lambda) u_1^2. 
\end{equation} 
By Lemma~\ref{lem:lambda}, we have $- \kappa_2 \lambda (1 - \lambda) >
0$ as long as $\kappa_1 > 0$ and $\kappa>0$. Thus, we have
\begin{equation}
  \average{u, H_0 u} \geq \kappa \sum_{j=0}^{\infty} \abs{ ( D u)_j}^2 - \kappa_2 \sum_{j=0}^{\infty} \abs{ (\Delta u)_j }^2
\end{equation}
Therefore, the scheme is stable as long as the underlying atomistic
model is stable.

For a general traction $T$, we have 
\begin{equation}
  \average{u, H_T u} = \average{u, H_0 u} - \frac{\kappa_2}{\kappa} (1- \lambda) T  u_1.
\end{equation}
Therefore, the stability follows from the stability in the case of $T
= 0$. 

\begin{remark}
This analysis also shows  that if $\lambda$ in \eqref{eq:tbc1d} is replaced by an appropriate approximation, i.e., $\widetilde{\lambda}\approx \lambda$ satisfying $-\kappa_2\wt{\lambda}(1-\wt{\lambda})>0$,  a stable model would also be obtained.
\end{remark}

\medskip

We show that a careless choice of the BC may lead to instability of the scheme. Instead of \eqref{eq:tbc1d}, let us consider
an alternative set of BCs (to distinguish, we use 
 $\wt{u}$ for the displacement)
\begin{subequations}\label{eq:tbc1d_unstable}
\begin{align}
  & \wt{u}_{-1} = (1+ \lambda^{-1}) \wt{u}_0 - \lambda^{-1} \wt{u}_1 + (\lambda^{-1} - 1) T / \kappa;  \\
  & \wt{u}_{-2} = (1+ \lambda^{-1}) \wt{u}_{-1} - \lambda^{-1} \wt{u}_0 +
  (\lambda^{-1} - 1) T / \kappa.
\end{align}
\end{subequations}
It is straightforward to check that this set of BC
also yields traction $T$ at the boundary
\begin{equation}
  - \kappa_1 (\wt{u}_0 - \wt{u}_{-1}) - \kappa_2 (\wt{u}_0 - \wt{u}_{-2}) - \kappa_2 (\wt{u}_1 - \wt{u}_{-1}) = T, 
\end{equation}
and hence consistent with the traction BC in continuum
elasticity. However, the resulting scheme with the BC
\eqref{eq:tbc1d_unstable} is not stable. In fact, we even lose uniqueness: it is easy to check that $\wt{u}_j = \lambda^j - 1$ for $j \geq -2$ satisfies
\begin{equation}
  - \kappa_2 (\wt{u}_{j+2} - 2 \wt{u}_j + \wt{u}_{j-2}) -
  \kappa_1 (\wt{u}_{j+1} - 2 \wt{u}_j + \wt{u}_{j-1}) = 0, \qquad j \geq 0,   
\end{equation}
and also at the boundary 
\begin{align}
  & \wt{u}_{-1} = (1+ \lambda^{-1}) \wt{u}_0 - \lambda^{-1} \wt{u}_1;   \\
  & \wt{u}_{-2} = (1+ \lambda^{-1}) \wt{u}_{-1} - \lambda^{-1} \wt{u}_0; \\
  & \wt{u}_0 = 0; \\
  & \limsup_{j \to \infty} \frac{\abs{\wt{u}_j}}{j} = 0. 
\end{align}

\subsection{Connection to BCs with applied forces at the boundary}

As we alluded to at the beginning of this section, it is also possible
to apply forces ($T_0$ and $T_1$) directly at the boundary to create a
traction. But it is not immediately clear how much forces to apply on
each of the two atoms at the boundary.  Here we will demonstrate the
connection to that approach. In particular, this discussion will also shed light on
the selection of the forces.

If we substitute \eqref{eq:tbc1d} into the force balance equation \eqref{eq:forcebalance}, we get 
\begin{align*}
  - V'(y_3 - y_1) - V'(y_2 - y_1) + V'(y_1 - y_0) & = f_1 + T_1; \\
  - V'(y_2 - y_0) - V'(y_1 - y_0) = f_0 + T_0.
\end{align*}
with 
\begin{align}
  T_1 & = \kappa_2 (u_{-1} - u_0) = \kappa_2 \bigl( (1 + \lambda) ( u_0 - u_1) + ( 1 - \lambda) T / \kappa \bigr); \\
  T_0 & = \kappa_2 (u_{-2} - u_0) + \kappa_1 (u_{-1} - u_0)  \\
  \notag & = \lambda \bigl( \kappa_1 + \kappa_2 ( 1 + \lambda) \bigr) (u_0 - u_1) + (1 - \lambda) \bigl( \kappa_1 + \kappa_2 ( 2 + \lambda) \bigr) T / \kappa.
\end{align}
This provides the formulas for the forces. An important observation,
however, is that these forces should depend on the displacement of the
atoms at $x_0$ and $x_1$.

\subsection{Traction BC from the Green's function}

To facilitate the extension of the BC to two
dimensional systems, we take yet another point of view of the traction
BC, from the lattice Green's function perspective.

Let us define the lattice Green's function associated with the model \eqref{eq:fdhalfchain},
\begin{equation}\label{eq:greens}
  - \kappa_2 (G_{j+2} - 2 G_j + G_{j-2}) -
  \kappa_1 (G_{j+1} - 2 G_j + G_{j-1}) = \delta(j), \qquad j \in \ZZ.
\end{equation}
In general, the lattice Green's functions are useful analytical tools
for studying lattice distortions around defects (see e.g.,
\cite{Tewary1973}). A typical route to compute the Green's function is
via a Fourier transform,
\begin{equation}
  \wh{G}(\xi) = \sum_{j \in \ZZ} e^{i \xi j} G_j, \qquad \xi \in [-\pi, \pi)
\end{equation}
with the inverse given by, 
\begin{equation}\label{eq:FourierInversion}
  G_j = \frac{1}{2\pi} \int_{-\pi}^{\pi} e^{-i \xi j} \wh{G}(\xi) \ud \xi
\end{equation}

This leads to 
\begin{equation}
  - \kappa_2 \bigl( e^{2i \xi} - 2 + e^{-2i \xi} \bigr) \wh{G}(\xi)
  - \kappa_1 \bigl( e^{i\xi} - 2 + e^{-i\xi} \bigr) \wh{G}(\xi) = 1,
\end{equation}
and
\begin{equation}
  \wh{G}(\xi) = \frac{1}{4 \kappa_2 \sin^2(\xi) + 4 \kappa_1 \sin^2(\xi / 2)}.
\end{equation}

However, due to the singularity at $\xi = 0$, the integral
\eqref{eq:FourierInversion} with $\wh{G}$ given above is not well
defined.  A remedy \cite{MaRo02} is to modify \eqref{eq:FourierInversion}:
\begin{equation}\label{eq:GreenFunction}
  \begin{aligned}
    G_j & = \frac{1}{2\pi} \int_{-\pi}^{\pi} \frac{e^{-i\xi j} - 1}{4
      \kappa_2 \sin^2(\xi) + 4 \kappa_1 \sin^2(\xi / 2)} \ud \xi \\
    & = - \frac{1}{2\pi} \int_{-\pi}^{\pi} \frac{2  \sin^2(\xi j/2)}{4
      \kappa_2 \sin^2(\xi) + 4 \kappa_1 \sin^2(\xi / 2)} \ud \xi. 
  \end{aligned}
\end{equation}
Conceptually, the Green's function
\eqref{eq:greens} is only defined up to a constant, and one can fix $G_0 =
0$ by subtracting a (infinite) constant from \eqref{eq:FourierInversion}.

As a result, the integral is now well defined as the integrand is regular as $\xi
\to 0$. The function $G_j$ defined this way still satisfies the equations  \eqref{eq:greens}.
We now make the connections to the BCs \eqref{eq:tbc1d}.

\begin{lemma} For $j \leq 0$, 
  \begin{equation}\label{eq:Gidentity}
    G_{j-1} = (1 + \lambda) G_j - \lambda G_{j+1}
  \end{equation}
\end{lemma}

\begin{proof}
  Rewrite \eqref{eq:GreenFunction} using a change of variable $z =
  \exp(i \xi)$ and the characteristic polynomial associated to the
  denominator, we get
  \begin{equation*}
    \begin{aligned}
      G_j & = \frac{1}{2\pi \kappa_2 i}\int_{\abs{z} = 1} \frac{(z^{-j} - 1)z}{ (z - 1)^2 (z - \lambda) (z - \lambda^{-1}) } \ud z \\
      & = \frac{1}{2\pi \kappa_2 i} \lim_{\eps \to 0}
      \int_{\gamma_{\eps}} \frac{(z^{-j} - 1)z}{ (z - 1)^2 (z -
        \lambda) (z - \lambda^{-1}) } \ud z
    \end{aligned}
  \end{equation*}
  where the contour $\gamma_{\eps}$ is given by the boundary of
  $B_1(0) \backslash B_{\eps}(1)$ on the complex plane. The second
  equality uses the fact that the integrand is regular as $z \to 1$.

  Using this representation, we have 
  \begin{equation*}
    \begin{aligned}
      G_{j-1} - (1 + \lambda) G_j + \lambda G_{j+1} & = \frac{1}{2\pi
        \kappa_2 i} \lim_{\eps \to 0} \int_{\gamma_{\eps}}
      \frac{z^{-j} \bigl(z^2 - (1 + \lambda) z + \lambda \bigr)}{ (z -
        1)^2 (z -
        \lambda) (z - \lambda^{-1}) } \ud z \\
      & = \frac{1}{2\pi \kappa_2 i} \lim_{\eps \to 0}
      \int_{\gamma_{\eps}} \frac{z^{-j}}{ (z - 1) (z - \lambda^{-1}) } \ud z.
    \end{aligned}
  \end{equation*}
  As $j \leq 0$ and $\abs{\lambda} < 1$, the integrand is holomorphic
  in $B_1(0) \backslash B_{\eps}(1)$ for any $\eps$, and hence the
  integral vanishes for any $\eps$ by Cauchy's integral
  theorem. Therefore, \eqref{eq:Gidentity} holds.
\end{proof}

The equation \eqref{eq:Gidentity} is exactly in the same form as the BCs
\eqref{eq:tbc1d} when $T=0$. This is not surprising, since the Green's function represents a special set of solutions. In particular, $G_j$ satisfies the homogeneous difference equations \eqref{eq:fdhalfchain}. Nevertheless, this simple observation can be employed to determine the coefficients in the BCs by using the Green's functions as test functions. This will be implemented for problems in two-dimensions, and the implementation will be discussed in the next sections.

\section{Implementation in two-dimensional models}\label{sec: 2d}

Here we demonstrate how the BC can be extended to two-dimensional systems.

\subsection{The traction BC and the induced boundary value problem}

For multi-dimensional problems, the BC is typically non-local \cite{MeKaLi06}, in that the displacement
of all the atoms at the boundary is coupled. It is also possible to consider nonlocal boundary
conditions, for example, in the spirit of the boundary element method
for molecular static models by one of the authors \cite{Li:12}.
Another alternative is to seek a {\it local} BC, in the sense that the position
of the ghost atoms are only determined by the positions of nearby
atoms in the system.  To
make the dependence local, we would employ a ``local flattening'' of
the boundary. Roughly speaking, for an atom at the boundary, the position is determined by a homogeneous
approximation of the local atom configuration with the local value of
the traction tensor.

To better explain the idea, we consider the face-centered cubic (FCC) lattice of Aluminum with the
axis aligned in $\langle 11 0\rangle$, $\langle 0 0 1\rangle$ and $\langle 1\bar{1} 0\rangle$ orientations. 
When projected to the $\langle 1\bar{1} 0\rangle$ plane, the lattice spacing in the horizontal and vertical directions are 
$\frac{a_0}{\sqrt{2}}$ and $\frac{a_0}{2}$, respectively, which makes it look like a triangular lattice,
as shown in Fig. \ref{fig: Al2d}. Again, we introduce ghost atoms outside the boundary in order to achieve the
desired traction condition. They are represented by open circles in Fig.  \ref{fig: Al2d}. 
\begin{figure}[htp]
\begin{center}
\scalebox{0.4}{\includegraphics{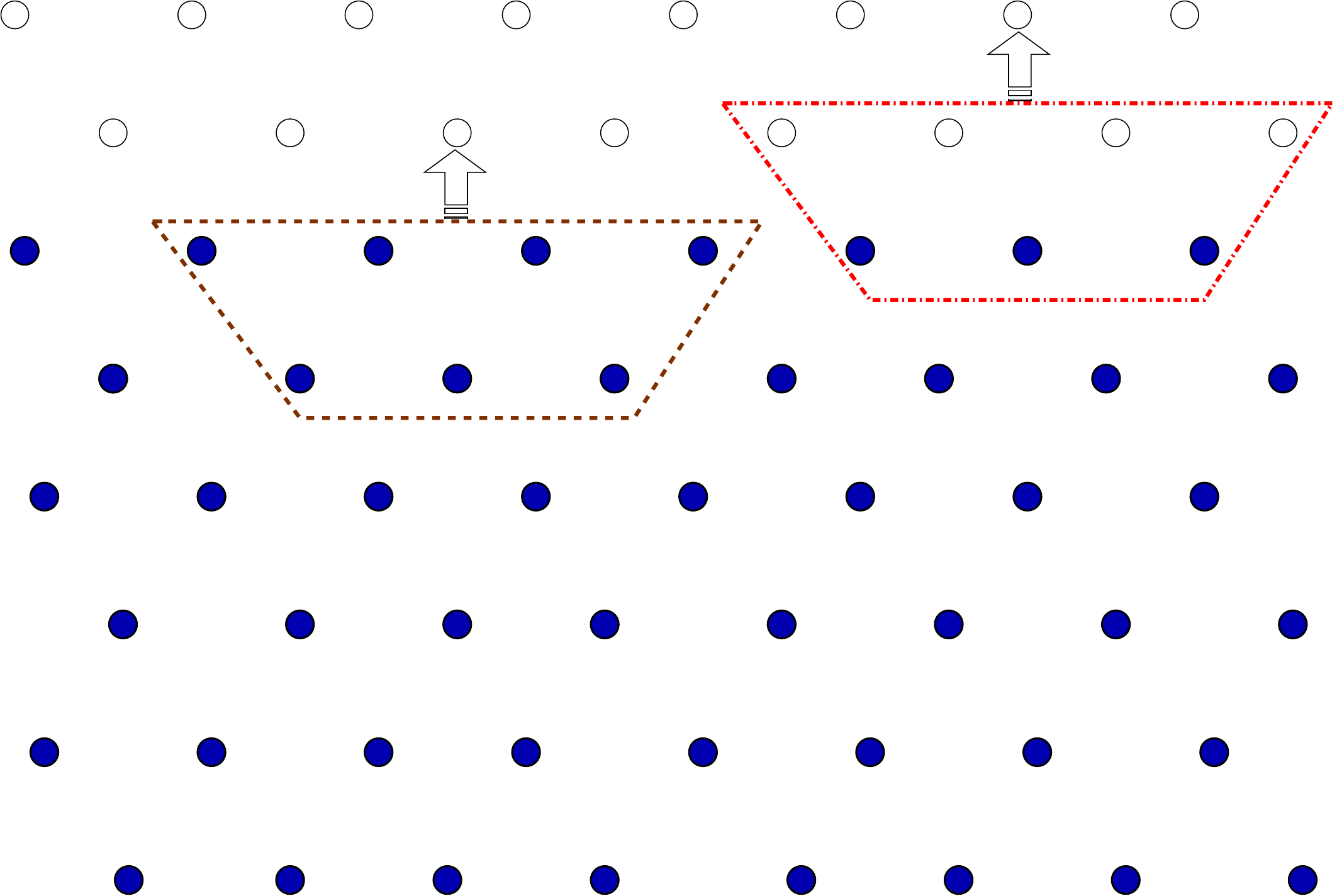}}
\end{center}
\caption{The projected atomic position of a FCC lattice. Filled circles: Atoms in the interior; Open circles: 
the ghost atoms introduced outside the boundary. The boxes contain the set of atoms $S_j$ that will be used to determine the displacement of the $j$th atom (see equation \eqref{eq: loc-bc-2d}). }
\label{fig: Al2d}
\end{figure}

Our main goal is to determine the actual position of the ghost atoms
based on the displacement of the atoms in the interior and the applied
traction $T$, which is a two-dimensional vector. In this case, it is
in general cumbersome to obtain the exact boundary
condition. Motivated by the one-dimensional traction BC, we seek an approximate BC in the following
form,
\begin{equation}\label{eq: loc-bc-2d}
  u_j = \sum_{i\in S_j} B_{ji} u_i + p_j.
\end{equation}

The shift vector $p$ is similar to the non-homogeneous term in
\eqref{eq:tbc1d}, and it will be determined so that the correct
traction is obtained. In the case when $p=0$, this boundary condition
would coincide with the BCs that models an environment that is at a
mechanical equilibrium with zero stress \cite{Li2009b,MeKaLi06}. In
principle, an exact BC in this form can be derived, e.g., in
\cite{MeKaLi06}, which mathematically, is a discrete analogue of the
Dirichlet-to-Neumann (DtN) map. The exact expression is typically
nonlocal, in that the summation is over all the atoms near the
boundary. But here we choose a local approximation, and restrict the
summation in \eqref{eq: loc-bc-2d} to those atoms that are close to
the $j$ atoms.  These neighbors are collected by the set $S_j.$ Due to
the translational symmetry of the lattice, we will use the same set of
neighbors when implementing the formula \eqref{eq: loc-bc-2d}, which
is also demonstrated in Fig. \ref{fig: Al2d}.  More specifically, we
start with the layer of ghost atoms closest to the boundary and apply
the BCs \eqref{eq: loc-bc-2d}. Once the displacement of these atoms
are updated, we move to the next layer, and these steps will be
repeated until the position of all the ghost atoms are updated.

We now discuss how to determine the coefficients $B_{ji}$. Since they are independent of the applied traction, they can be computed for the case $T=0$ and $p=0.$ In this case, these coefficients can be determined
using an optimization procedure, developed in \cite{Li2009b}. More specifically, we choose an objective function as follows, 
\begin{equation}
 \min h, \quad  h= \sum_{k} |G_{j,k} - \sum_i B_{ji} G_{i,k}|^2.
\end{equation}
Here $G_{j,k}$ is the two-dimensional lattice Green's function  \cite{Tewary1973}. 
The main observation is that the BC should by satisfied by special solutions, especially the lattice Green's functions $G_{i,k}$, which corresponds to the solution of the linearized molecular statics model when a point force is applied on the $k$th atom. This was already observed for the one-dimensional model. Ideally, the objective function would be zero when the BC is exact.  Further, we introduce a constraint,  
\begin{equation}
 \sum_i B_{ji}=I,
\end{equation}
so that the constant solutions are admitted. This is also seen in the
one-dimensional system: The two coefficients in \eqref{eq:tbc1d} add
up to 1.

It remains to estimate the vector $p$. In principle, it should be determined by requiring the traction to arrive at a prescribed value. The total tractions along the boundary is given by the sum of the forces \cite{Admal2010unified,WuLi14},
\begin{equation}
  t= \sum_{i\in \Omega, j\notin \Omega} f_{ij}.
\end{equation}
Here, $f_{ij}$ comes from a force decomposition. Namely, 
\begin{equation}\label{eq: fij}
  f_i= \sum_{j\ne i} f_{ij},  \quad f_{ij} = -f_{ji}.  
\end{equation}
This formula, which was already indicated by \eqref{eq:tractioncondition}, is consistent with the intuition of Cauchy. The explicit expressions  for the force decomposition \eqref{eq: fij}, especially for multi-body potentials, can be found in \cite{Chen2006local,WuLi14}.

To control the local traction, we divide the region with ghost atoms into blocks, each denoted by $\Omega^c_\alpha,$
$\alpha=1, 2, \cdots, M$. The computational domain is denoted by $\Omega,$ and the intersection with $\Omega^c_\alpha$, is written as $\partial \Omega_\alpha.$ For each block $\Omega^c_\alpha,$ we introduce a shift vector $p_\alpha$. They are chosen so that the traction along $\partial \Omega_\alpha$ agrees with a prescribed value $t_\alpha$. This arrangement is illustrated in Fig. \ref{fig: blocks}.
\begin{figure}[htp]
\begin{center}
\scalebox{0.4}{\includegraphics{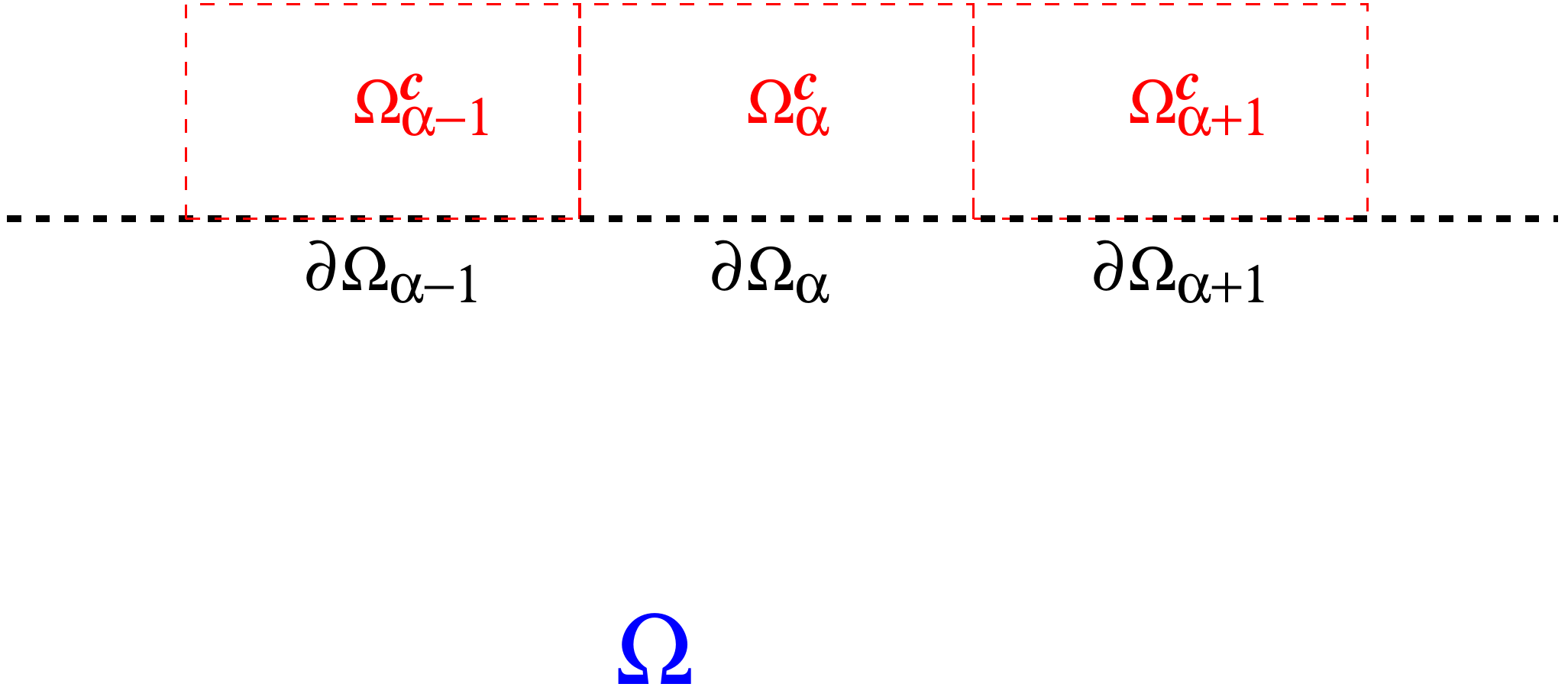}}
\end{center}
\caption{Imposing tractions along the boundary: $\Omega$ indicate the computational domain. The atoms outside the boundary are grouped into blocks $\Omega_\alpha^c$ with the intersection with the boundary given by $\partial\Omega_\alpha$. The traction on each  $\partial\Omega_\alpha$ is prescribed. }
\label{fig: blocks}
\end{figure}

We now put the mathematical models together. 
\begin{equation}\label{eq: model2d}
 \left\{
 \begin{aligned}
  & \frac{\partial}{\partial u_i} V=0,  &&\forall x_i \in \Omega, \\
  & u_j = \sum_{i} B_{ji} u_i + p_\alpha, & &\forall x_j \in \partial \Omega_\alpha,\\
  & \sum_{i\in \Omega, j\in \Omega_\alpha^c} f_{ij} = t_\alpha.&
  \end{aligned}\right.
\end{equation} 
The first set of equations represent the force balance in the interior, with potential energy given by $V$. The remaining equations serve as BCs with prescribed tractions $t_\alpha$.  The unknowns are the atomic displacement, together with the shift vectors $p_\alpha$.
The atomic degrees of freedom associated with the atoms outside $\Omega$ has been implicitly taken into account by
the second and third equations in \eqref{eq: model2d}. In the next section, we will discuss 
an implementation method.

 \subsection{Numerical implementations}

Our reduced model \eqref{eq: model2d} consists of a set of nonlinear algebraic equations. It is therefore natural
to make use of iterative methods, such as the quasi-Newton's method. In general, this requires the approximation
of the Jacobian matrix, since the analytical form is usually not available. 
The convergence is typically slow, especially when the system is not well prepared. 

To find an alternative, we notice that in the domain $\Omega$, the molecular statics model is associated with 
an energy. Therefore, for Dirichlet BCs, where the atoms outside the boundary are held fixed, 
the solution to the molecular statics model correspond to an energy minimization, which is usually more robust and
much more efficient than solving the nonlinear equations. 
 
We implement the equations by a domain decomposition approach, and
alternate among the three sets of equations in \eqref{eq: model2d}. As
an example to explain the idea, we may consider the coupling of the
first two sets of equations and assume that $p_\alpha$ is given.  We
create a few overlapping layers, in which the atoms serve as both
ghost atoms and interior atoms.  This is illustrated by Fig. \ref{fig:
  DD2d}. In each iteration, we first update the displacement of all
the ghost atoms including those in the overlapping region using
\eqref{eq: loc-bc-2d} (or the second equation in \eqref{eq: model2d}).
We then turn to the interior atoms, assuming that other atoms are held
fixed (open circles in Fig. \ref{fig: DD2d}). By minimizing the
energy, we obtain the updated position of the interior atoms. The
numerical implementation has been done with the BFGS package
\cite{Liu1989lbfgs}. This iteration can be continued until convergence
is reached. This is simply the Schwartz iteration
\cite{toselli2005domain}) between the two models.
 
 \begin{figure}[htp]
 \begin{center}
\scalebox{0.35}{\includegraphics{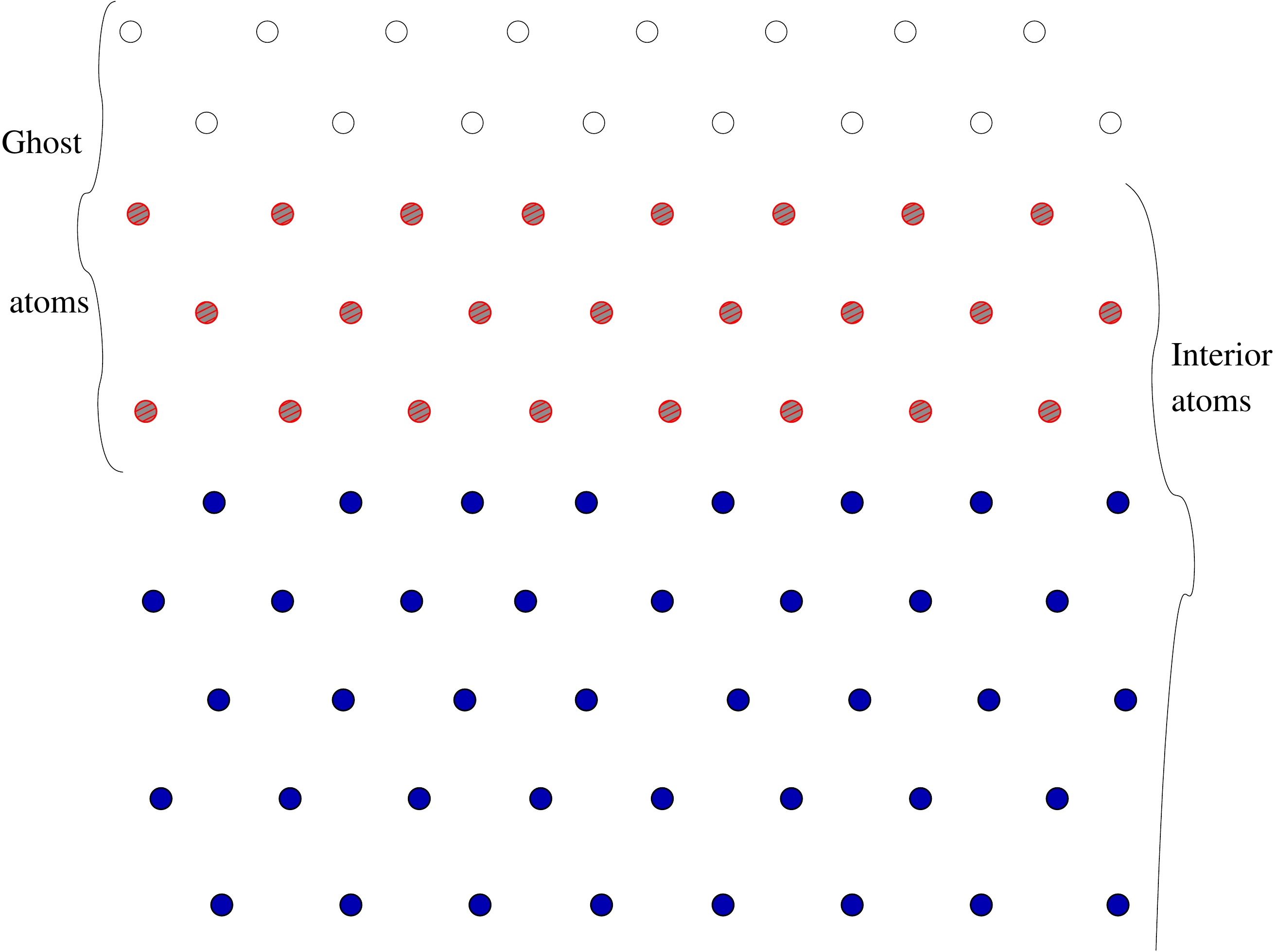}}
\caption{An overlapping domain decomposition.}
\label{fig: DD2d}
\end{center}
\end{figure}

\subsection{Results from numerical experiments}

As a test problem, we consider a dislocation dipole under a shear load. The atoms around the two dislocations with opposite Burgers vectors are shown in
Fig. \ref{fig: dpl}. The embedded atom model (EAM) \cite{ercolessi1994interatomic} has been used as the interatomic potential. 
\begin{figure}[htpb]
\begin{center}
\scalebox{0.45}{\includegraphics{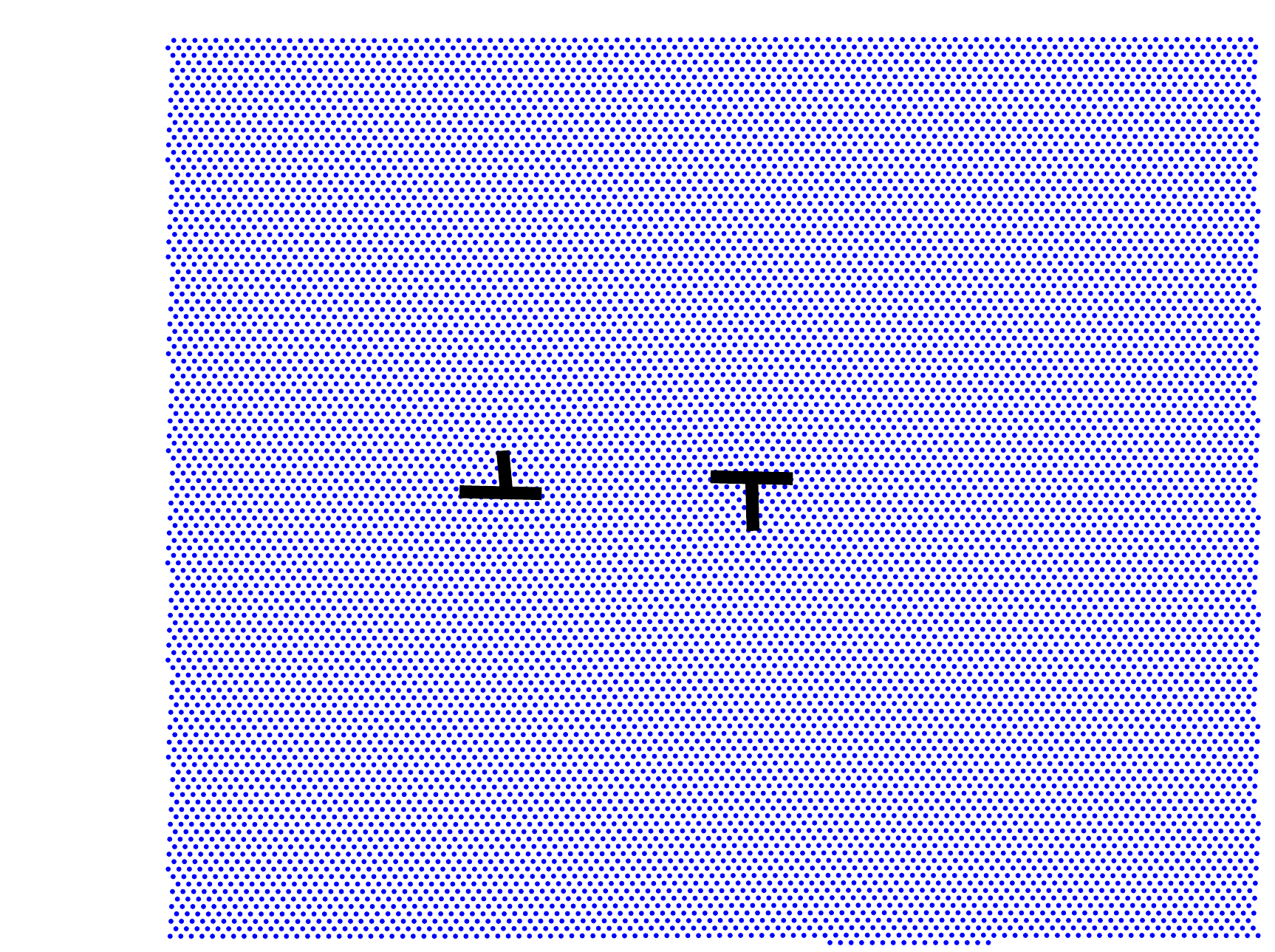}}
\caption{The atoms near the two dislocations.}
\label{fig: dpl}
\end{center}
\end{figure}

We first manually control the vector $p=(p_1,0),$ and observe the influence on 
the traction changes.  As a quasi-static loading, we increase $p_1$ with small increment and
then solve the molecular statics model using the domain decomposition method described in the
previous section. For each step, we also compute the traction
at the boundary. The history of the total boundary traction is shown in Fig. \ref{fig: thist}. We observe that the traction increases as $p_1$ increases. However, when $p_1$ reaches certain critical value,  a sudden drop is observed. In this case, the two 
dislocations move to the boundary, and the entire sample undergoes a complete slip. Fig. \ref{fig: positions} shows the atomic positions before and after the slip.
\begin{figure}[hpbt]
\begin{center}
\scalebox{0.45}{\includegraphics{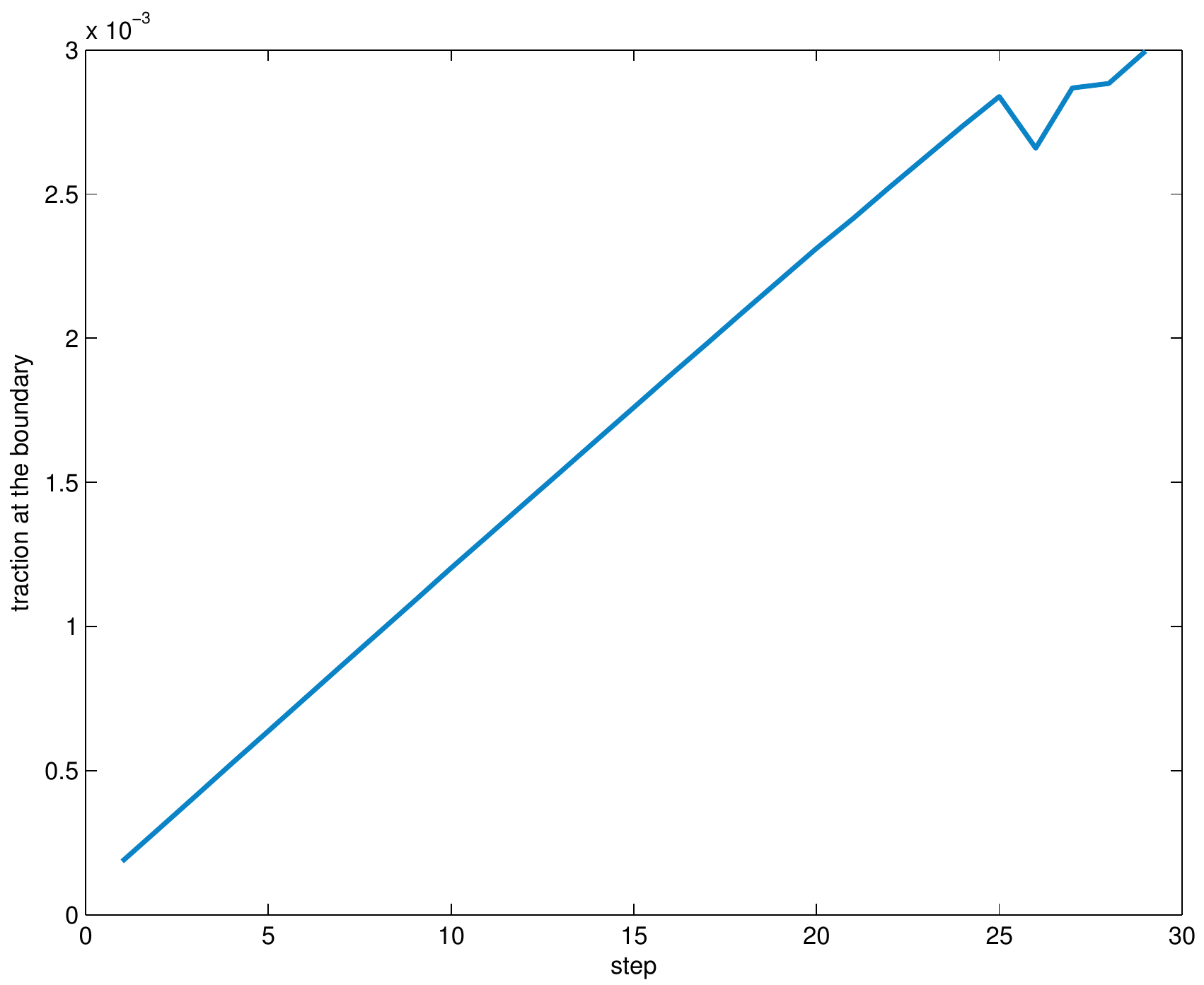}}
\caption{The history of the traction at the boundary.}
\label{fig: thist}
\end{center}
\end{figure}

\begin{figure}[htbp]
\begin{center}
\scalebox{0.85}{\includegraphics{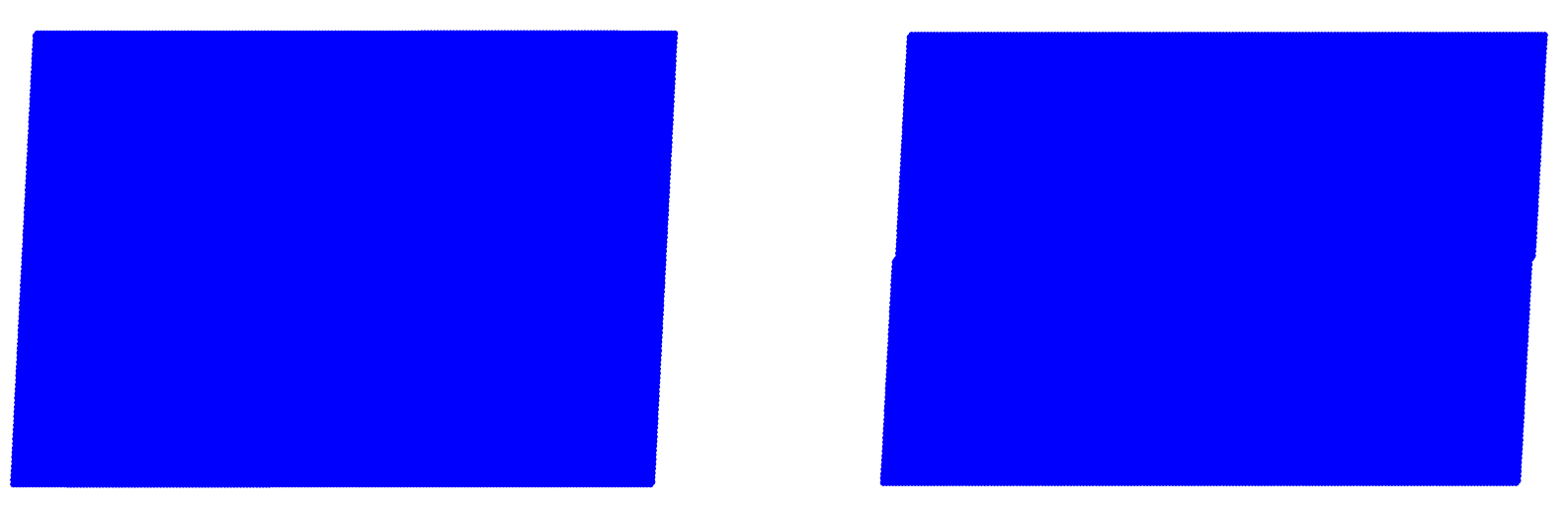}}
\caption{The position of the atoms before and after the yield stress.}
\label{fig: positions}
\end{center}
\end{figure}

In the next experiment, we apply a uniform traction along the upper and lower boundaries. 
In Fig. \ref{fig: tpvals}, we show the resulting values of $p_1$ along the boundary. While, the resulting tractions
have reached the prescribed values $(2\times 10^{-4},0)$, it is clear that
the values for $p$ are not homogeneous, mainly due to the presence of the dislocations. The displacement is shown in Fig. \ref{fig: u1}, together with a close-up view of the atomic positions. All these results suggest that the atomic positions are not uniform.
Compared to the simulations of dislocation dipoles  using the Parrinello-Rahman method (e.g.,\cite{wang2004calculating}),
the current approach does not introduce periodic images of the dislocation dipole. Moreover, the uniform shear stress can be applied without forcing a uniform deformation along the boundaries. 

\begin{figure}[htpb]
\begin{center}
{\includegraphics[width=5in,height=2.2in]{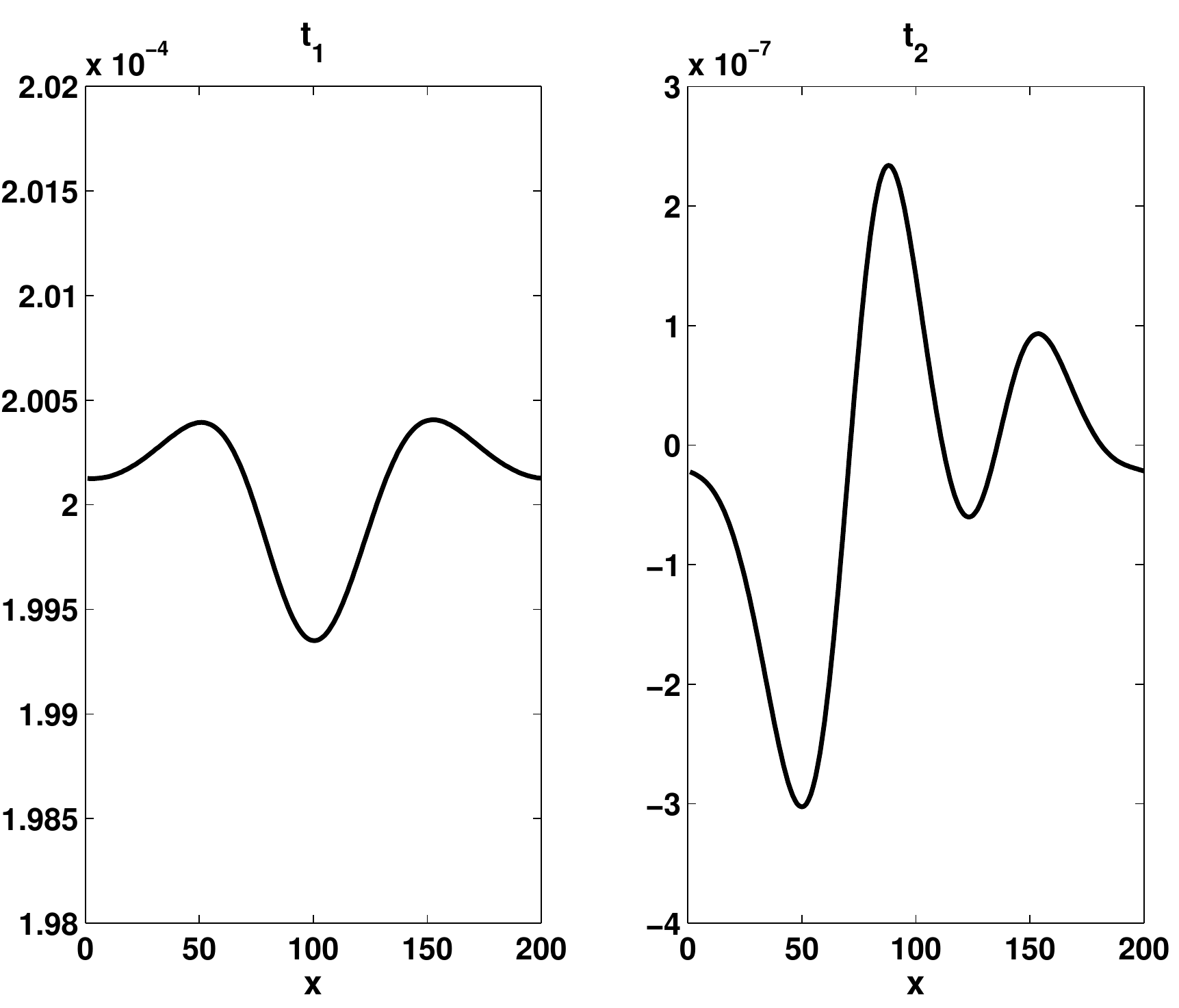}}\\
{\includegraphics[width=5in,height=2.2in]{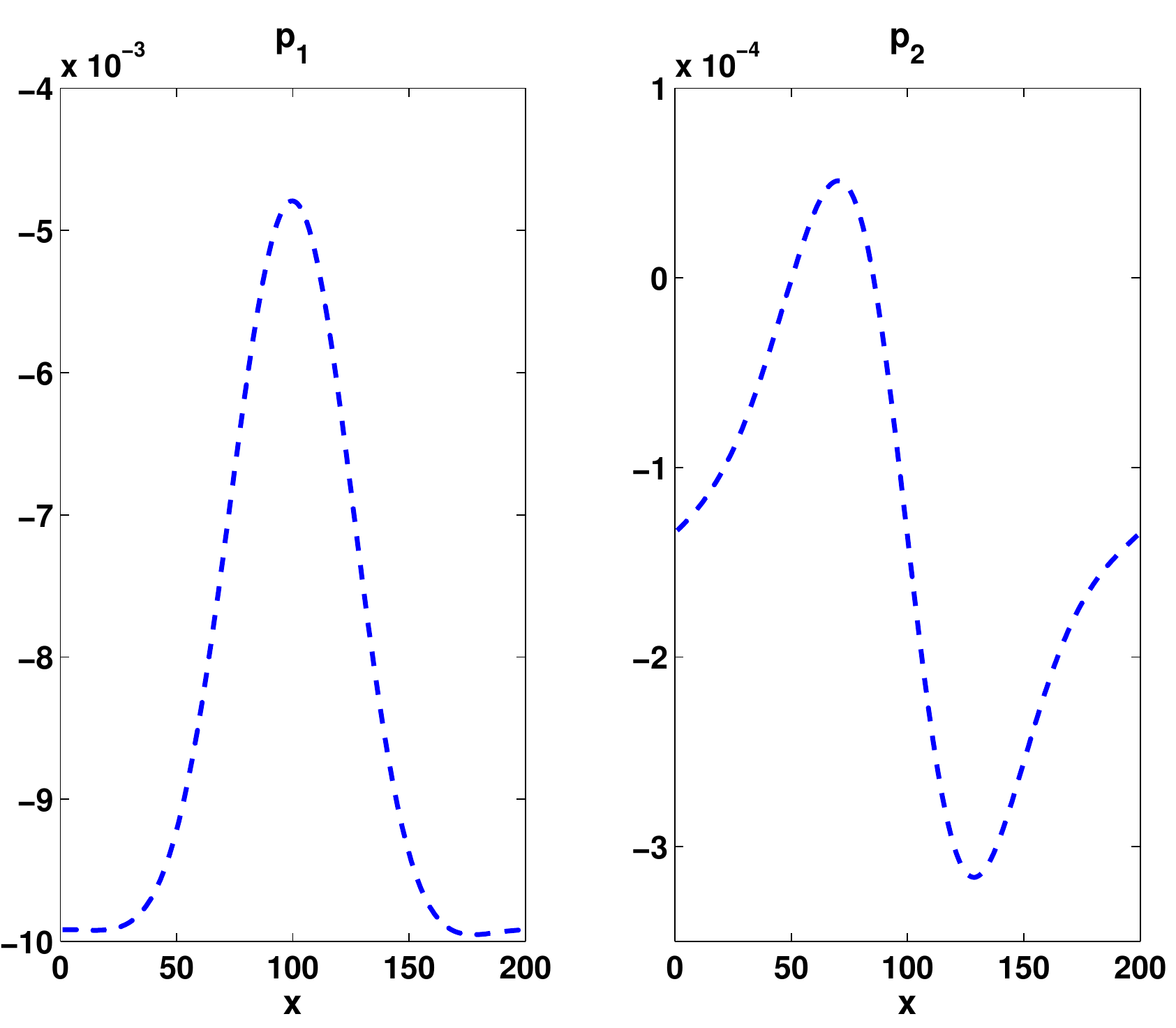}}
\caption{The traction (top) and the shift vector $p_1$ (bottom) along the upper boundary.}
\label{fig: tpvals}
\end{center}
\end{figure}

\begin{figure}[hbtp]
\begin{center}
{\includegraphics[scale=0.17]{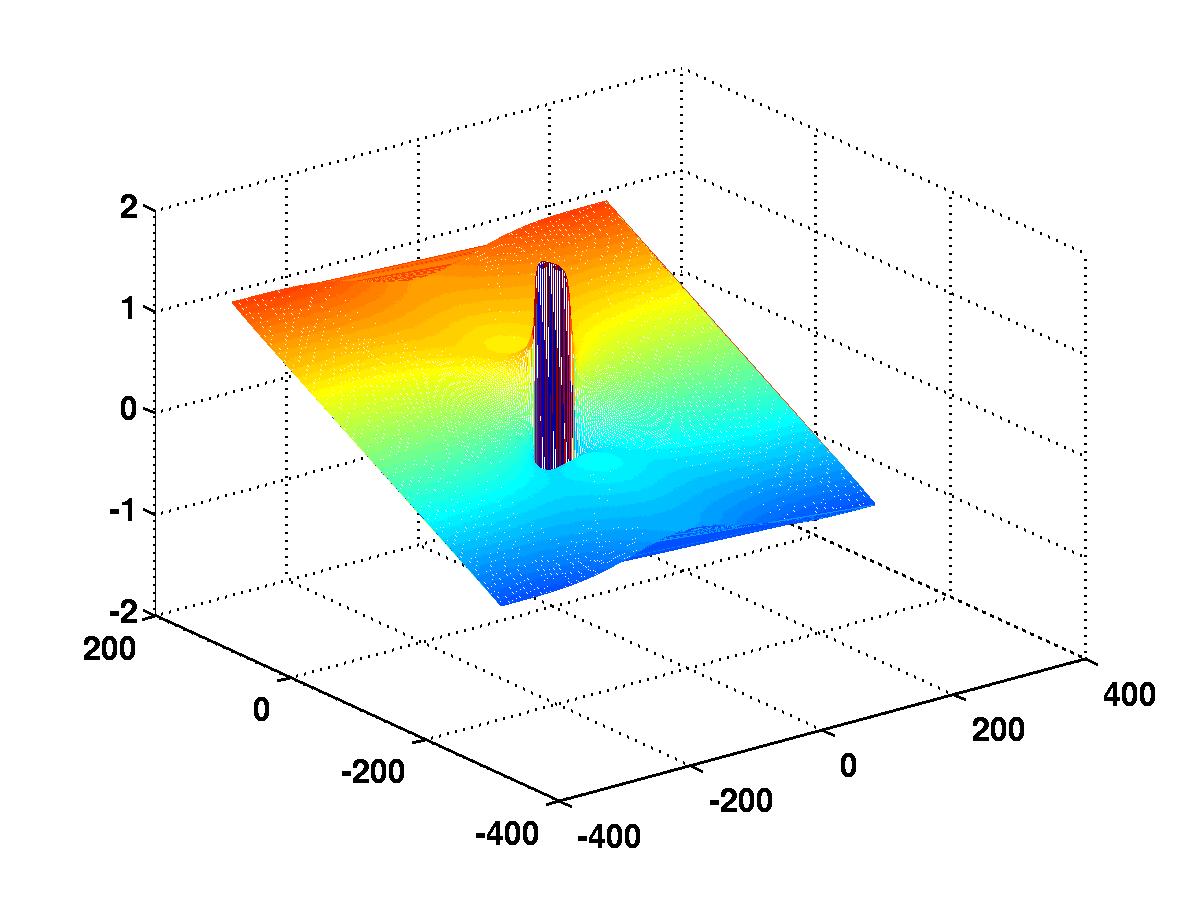}}
{\includegraphics[scale=0.32]{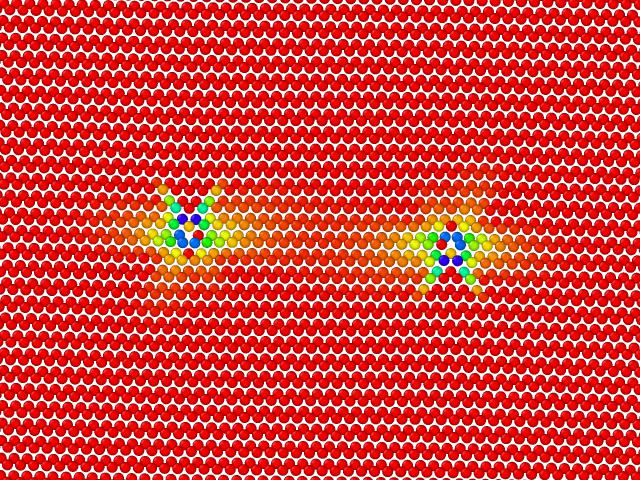}}
\caption{The displacement $u_1$ and a close-up view of the atoms near the two dislocations (generated in Ovito \cite{stukowski2010visualization}).}
\label{fig: u1}
\end{center}
\end{figure}

\section{Summary and Discussion}
We have formulated boundary conditions that impose a traction force on
a molecular statics model.  These boundary conditions are derived by
taking into account the surrounding elastic environment. Hence, the
computational domain is part of a much bigger sample, and artificial
boundary effects can be eliminated. In the continuum limit, these
boundary conditions coincide with the Neumann boundary condition for
continuum elasticity models.  We have restricted our discussions to
static problems. Extension to dynamic problems will be considered in
future works.

\medskip

\noindent\textbf{Acknowledgment.} 
The work of J.L. was supported in part by the Alfred P. Sloan
foundation and National Science Foundation under award DMS-1312659.

\bibliographystyle{amsxport}
\bibliography{tbc}

\end{document}